\documentclass[12pt,a4paper]{article}
\usepackage[margin=1in]{geometry} %In Folder
\usepackage{amsmath,amssymb,amsthm}
\usepackage{color}
\usepackage{graphicx}
\usepackage{subfigure}
\usepackage{booktabs}
\usepackage{threeparttable}
\usepackage{multirow}
\usepackage{array}
\usepackage{amsmath,amsthm,amssymb,amscd}
\usepackage{latexsym}
\usepackage{indentfirst}
\usepackage{bibentry}
\usepackage{textcomp}
\usepackage{float}
\usepackage{multirow}
\usepackage{booktabs}
\usepackage{graphicx}
\usepackage{color}
\usepackage{cite}
\usepackage{authblk}
\usepackage{algorithmicx,algorithm}
\usepackage{listings}
\usepackage{xcolor}
\usepackage{alltt}
\usepackage{enumitem}
\usepackage{lscape}
\usepackage{geometry}
\usepackage{enumitem}
\usepackage{longtable}
\usepackage{afterpage}
\usepackage{capt-of}
\usepackage{extarrows}
\usepackage{mathrsfs}
\usepackage{rotating}
\newtheorem{theorem}{Theorem}[section]

\newtheorem{proposition}[theorem]{Proposition}
\newtheorem{corollary}[theorem]{Corollary}

\newtheorem*{rem}{Remark}

\newtheorem{remark}{Remark}

\numberwithin{equation}{section}

\newcommand{\be}{\begin{equation}}
\newcommand{\ee}{\end{equation}}

\newcommand{\bea}{\begin{eqnarray}}
\newcommand{\eea}{\end{eqnarray}}
\newcommand{\bean}{\begin{eqnarray*}}
\newcommand{\eean}{\end{eqnarray*}}

\newcommand{\bg}{\begin{gather}}
\newcommand{\eg}{\end{gather}}
\newcommand{\bgn}{\begin{gather*}}
\newcommand{\egn}{\end{gather*}}

\title{Painlev\'{e} IV, $\sigma-$Form and the Deformed Hermite Unitary Ensembles}
\author[a,b]{{Mengkun Zhu}\footnote{Zhu\_mengkun@163.com}}
\author[b]{{Dan Wang}\footnote{Corresponding author: bohewan@126.com}}
\author[b]{{Yang Chen}\footnote{yangbrookchen@yahoo.co.uk}}
\affil[a]{\small School of Mathematics and Statistics, Qilu University of Technology (Shandong Academy of Sciences),
Jinan 250353, China}

\affil[b]{Department of Mathematics, Faculty of Science and Technology, University of Macau,
Avenida da Universidade, Taipa, Macau, China}

\date{}
\begin{document}
\maketitle

\abstract{\noindent We study the Hankel determinant generated by a deformed Hermite weight with one jump $w(z,t,\gamma)=e^{-z^2+tz}|z-t|^{\gamma}(A+B\theta(z-t))$, where $A\geq 0$, $A+B\geq 0$, $t\in\textbf{R}$, $\gamma>-1$ and $z\in\textbf{R}$. By using the ladder operators for the corresponding monic orthogonal polynomials, and their relative compatibility conditions, we obtain a series of difference and differential equations to describe the relations among $\alpha_n$, $\beta_n$, $R_n(t)$ and $r_n(t)$. Especially, we find that the auxiliary quantities $R_n(t)$ and $r_n(t)$ satisfy the coupled Riccati equations, and $R_n(t)$ satisfies a particular Painlev\'{e} IV equation. Based on above results, we show that $\sigma_n(t)$ and $\hat{\sigma}_n(t)$, two quantities related to the Hankel determinant and $R_n(t)$, satisfy the continuous and discrete $\sigma-$form equations, respectively. In the end, we also discuss the large $n$ asymptotic behavior of $R_n(t)$, which produce the expansion of the logarithmic of the Hankel determinant and the asymptotic of the second order differential equation of the monic orthogonal polynomials.
}

\vspace{0.1cm}
\noindent
Keywords: Ladder operators, Hermite unitary ensembles, $\sigma-$form, Painlev\'{e} IV.  \\
MSC: 15B52, 42C05, 33E17

\section{Introduction}
\noindent The joint probability density function of the eigenvalues $\{z_i\}_{i=1}^{n}$,
\begin{equation*}
p(z_1, z_2,\cdots,z_n)=\frac{1}{n!D_n(w_0)}\prod\limits_{1\leq i<k\leq n}(z_k-z_i)^2\prod\limits_{j=1}^nw_0(z_j),
\end{equation*}
is a well known fact in the theory of random matrix ensembles \cite{CHM2013,CL1998,CZ2010,M2006},
where $w_0(z)$ is a weight function on the interval $(a,b)$ and has the finite moments, i. e.
\begin{equation*}
\mu_k:=\int_a^bz^kw_0(z)dz, ~~~~k\in\{0,1,2,\cdots\}.
\end{equation*}
Here $D_n(w_0)$ is the normalization constant
\begin{align}\label{1.2}
D_n(w_0)=\frac{1}{n!}\int_{[a,b]^n}\prod\limits_{1\leq i<k\leq n}(z_k-z_i)^2\prod\limits_{j=1}^nw_0(z_j)dz_j,
\end{align}
so that
\begin{equation*}
\int_{[a,b]^n}p(z_1,z_2,\cdots,z_n)dz_1dz_2\cdots dz_n=1.
\end{equation*}

In this paper, we consider
\begin{align}\label{1.1}
w(z,t,\gamma)=e^{-z^2+tz}|z-t|^{\gamma}(A+B\theta(z-t)),
\end{align}
where $A\geq 0$, $A+B\geq 0$, $t\in\textbf{R}$, $\gamma>-1$ and $z\in\textbf{R}$,
which
\begin{equation*}
w_0(z)=e^{-z^2+tz}, ~~t\in\textbf{R}, ~~z\in\textbf{R},
\end{equation*}
corresponds to the deformed Hermite (or Gaussian) unitary ensemble.
Here $\theta(x)$ is the Heaviside function, i.e.
$$ \theta(x)=\left\{
\begin{aligned}
1&,~~ x>0 \\
0&,~~ x\leq 0
\end{aligned}
\right..
$$
According to the general theory of orthogonal polynomials of one variable, the normalization constant (\ref{1.2}) has the two more alternative representation
\begin{align*}
D_n(w)=\det(\mu_{i+j}(t))_{i,j=0}^{n-1}=\det\left(\int_{\textbf{R}}z^{i+j}w(z,t)dz\right)_{i,j=0}^{n-1}=\prod\limits_{k=0}^{n-1}h_k(t),
\end{align*}
where the determinant of the moment matrix $\mu_{i+j}(t)$ is the Hankel determinant, and $\{h_k(t)\}_{k=0}^n$ is the square of the $L^2$ norm of the sequence of polynomials $\{P_k(z)\}_{k=0}^n$.

Min and Chen \cite{MC2019} have studied the Painlev\'{e} transcendents and the Hankel determinants generated by a discontinuous Gaussian weight
\begin{equation}\label{1.10}
w(x,t_1, t_2)=e^{-x^2}(A+B_1\theta(x-t_1)+B_2\theta(x-t_2)),
\end{equation}
where $A_1$, $B_1$ and $B_2$ are constants, $A\geq 0$, $A+B_1\geq 0$, $A+B_1+B_2\geq 0$, $t_1< t_2$, $x\in\textbf{R}$. They considered the Gaussian weight with a single jump ($B_1$ or $B_2$ $=0$) or two jumps ($B_1\neq 0$ and $B_2\neq 0$), and proved their auxiliary quantities satisfy the second order difference and differential equations, respectively, via the ladder operators and supplementary conditions.

\indent
Let $P_n(z,t,\gamma)$ be the monic polynomials of degree $n$ orthogonal with respect to (\ref{1.1}),
\begin{equation}\label{1.3}
\int_{\textbf{R}}P_j(z,t,\gamma)P_i(z,t,\gamma)w(z,t,\gamma)dz=h_j(t,\gamma)\delta_{ji},
\end{equation}
where $j,i\in\{0,1,2,\cdots\}$, and $\delta_{ji}$ denotes the Kronecker delta.
It follows from the orthogonality relations that
\begin{equation}\label{1.4}
zP_n(z,t,\gamma)=P_{n+1}(z,t,\gamma)+\alpha_nP_n(z,t,\gamma)+\beta_nP_{n-1}(z,t,\gamma),~~n\geq 0,
\end{equation}
together with the initial conditions $$P_0(z,t,\gamma):=1,~~~ \beta_0P_{-1}(z,t,\gamma):=0.$$
The monic polynomials $P_n(z,t)$ have the monomial expansion
\begin{align}\label{1.5}
P_n(z,t,\gamma)=z^n+\textbf{p}(n,t,\gamma)z^{n-1}+\cdots+P_n(0,t,\gamma),
\end{align}
and we will see that the coefficient of $z^{n-1}$, $\textbf{p}(n,t,\gamma)$, play a significant role in the following discussions.
Note that due to the $t$ dependence of the weight, the coefficients of the polynomials and the recurrence coefficients $\alpha_n$ and $\beta_n$ also depend on $t$, the position of the jump.

From (\ref{1.4}) and (\ref{1.5}), for $n\in\{0,1,2,\cdots\}$,
\begin{equation}\label{1.6}
\alpha_n=\textbf{p}(n,t,\gamma)-\textbf{p}(n+1,t,\gamma),
\end{equation}
and
\begin{equation}\label{1.7}
\beta_n=\frac{h_n(t,\gamma)}{h_{n-1}(t,\gamma)}.
\end{equation}
A telescopic sum of (\ref{1.6}), it yields
\begin{equation}\label{1.8}
\sum\limits_{j=0}^{n-1}\alpha_j=-\textbf{p}(n,t,\gamma).
\end{equation}

The paper is organized up as follow. In section 2, we derive the ladder operators and the compatibility conditions ($S_1$), ($S_2$) and ($S_2^{\prime}$) with respect to the weight (\ref{1.1}). In section 3, we obtain some important identities on the auxiliary quantities $R_n(t)$ and $r_n(t)$ via applying the ladder operators. In section 4, we are interested in the parameter of the weight (\ref{1.1}), and we show that $r_n(t)$ and $R_n(t)$ satisfy the coupled Riccati equations.
 Section 5 is devoted to discuss two quantities $\sigma_n(t)$ and $\hat{\sigma}_n(t)$, which are allied to the Hankel determinant and $R_n(t)$. We prove that they satisfy the continuous or discrete $\sigma-$form equations, respectively. The final section, we make use of the equations in section $2-5$ to show that
 $R_n(t)$ satisfy a particular Painlev\'{e} IV equation. Based on this equation, we obtain the large $n$ asymptotic of the auxiliary quantity $R_n(t)$. Moreover, the logarithmic of the Hankel determinant $D_n(t)$ and the asymptotic of the second order differential equation of $P_n(z,t,\gamma)$ are also obtained.

\section{Ladder operators and compatibility conditions}
\noindent
A detailed proof of the lowering ladder operator, which is included in the ladder operators \cite{CP2005}, is shown in the following theorem.
For convenience,  we denote
\begin{align*}
w(z,t,\gamma)&=e^{-z^2+tz}|z-t|^{\gamma}(A+B\theta(z-t)),\\
w_0(z,t)&=e^{-z^2+tz},\\
w_{J}(z,t)&=A+B\theta(z-t),
\end{align*}
where $A\geq 0$, $A+B\geq 0$, $t\in\textbf{R}$, $\gamma>-1$ and $z\in\textbf{R}$. Meanwhile,
we do not always write down the $t$ or $\gamma$ dependence in $P_n(z,t,\gamma)$, $w(z,t,\gamma)$, $w_0(z,t)$, $w_{J}(z,t)$, $h_n(t,\gamma)$, $\alpha_n(t,\gamma)$ and $\beta_n(t,\gamma)$ unless it is needed.

\begin{theorem}\label{T1}
The monic orthogonal polynomials with respect to the weight (\ref{1.1}) satisfy the following differentiation formula:
\begin{equation}\label{2.1}
P_n^{\prime}(z)=\beta_nA_n(z)P_{n-1}(z)-B_n(z)P_n(z),
\end{equation}
where
\begin{align}\label{2.2}
A_n(z)=\frac{1}{h_n}\int_{\textbf{R}}P_n^2(y)\frac{v_0^{\prime}(z)-v_0^{\prime}(y)}{z-y}w(y)dy+a_n(t),
\end{align}
\begin{equation}\label{2.3}
B_n(z)=\frac{1}{h_{n-1}}\int_{\textbf{R}}P_n(y)P_{n-1}(y)\frac{v_0^{\prime}(z)-v_0^{\prime}(y)}{z-y}w(y)dy+b_n(t),
\end{equation}
\begin{equation}\label{2.4}
a_n(t)=\frac{\gamma}{h_n}\int_{\textbf{R}}\frac{P_n^2(y)}{(z-y)(y-t)}w(y)dy
\end{equation}
and
\begin{equation}\label{2.5}
b_n(t)=\frac{\gamma}{h_{n-1}}\int_{\textbf{R}}\frac{P_n(y)P_{n-1}(y)}{(z-y)(y-t)}w(y)dy.
\end{equation}
Here
\begin{equation*}
v_0(z)=-\ln w_0(z).
\end{equation*}
\end{theorem}
\begin{proof}
Firstly, we note \cite{CF2006}
\begin{equation}\label{2.5.1}
\frac{\partial(|z-t|^{\gamma})}{\partial z}=\bigg((z-t)^{\gamma}-(t-z)^{\gamma}\bigg)\delta(z-t)+\frac{\gamma|z-t|^{\gamma}}{z-t},
\end{equation}
which is obtained by writing
\begin{equation*}
|z-t|^{\gamma}=(z-t)^{\gamma}\theta(z-t)+(t-z)^{\gamma}\theta(t-z).
\end{equation*}
From
\begin{equation*}
P_n^{\prime}(z)=\sum\limits_{k=0}^{n-1}C_{n,k}P_k(z),
\end{equation*}
where $C_{n,k}$ is determined from the orthogonality relations
\begin{align*}
C_{n,k}&=\frac{1}{h_k}\int_{\textbf{R}}P_n^{\prime}(y)P_k(y)\omega(y)dy.
\end{align*}

Using the integration by parts and noting the weight (\ref{1.1}) vanishes at the endpoints of the orthogonality interval, we find
\begin{align*}
C_{n,k}&=\frac{1}{h_k}\int_{\textbf{R}}P_n^{\prime}(y)P_k(y)w(y)dy=-\frac{1}{h_k}\int_{\textbf{R}}P_n(y)P_k(y)w^{\prime}(y)dy\\
&=-\frac{1}{h_k}\int_{\textbf{R}}P_n(y)P_k(y)(v_0^{\prime}(z)-v_0^{\prime}(y))w(y)dy-\frac{\gamma}{h_k}\int_{\textbf{R}}P_n(y)P_k(y)\frac{w(y)}{y-t}dy;
\end{align*}
straightforward,
\begin{align}\label{2.6}
&P_n^{\prime}(z)\nonumber\\
=&\sum\limits_{k=0}^{n-1}\bigg(-\frac{1}{h_k}\int_{\textbf{R}}P_n(y)P_k(y)(v_0^{\prime}(z)-v_0^{\prime}(y))w(y)dy-\frac{\gamma}{h_k}\int_{\textbf{R}}P_n(y)P_k(y)\frac{w(y)}{y-t}dy\bigg)P_k(z)\nonumber\\
=&-\int_{\textbf{R}}P_n(y)\sum\limits_{k=0}^{n-1}\frac{P_k(y)P_k(z)}{h_k}(v_0^{\prime}(z)-v_0^{\prime}(y))w(y)dy-\gamma\int_{\textbf{R}}P_n(y)\sum\limits_{k=0}^{n-1}\frac{P_k(y)P_k(z)}{h_k}\frac{w(y)}{y-t}dy\nonumber\\
=&-\frac{1}{h_{n-1}}\int_{\textbf{R}}\frac{v_0^{\prime}(z)-v_0^{\prime}(y)}{z-y}P_n(y)P_{n-1}(y)w(y)dyP_n(z)\nonumber\\
&+\frac{1}{h_{n-1}}\int_{\textbf{R}}\frac{v_0^{\prime}(z)-v_0^{\prime}(y)}{z-y}P_n^2(y)w(y)dyP_{n-1}(z)\nonumber\\
&-\frac{\gamma}{h_{n-1}}\int_{\textbf{R}}P_n(y)P_{n-1}(y)\frac{w(y)}{(z-y)(y-t)}dyP_n(z)\nonumber\\
&+\frac{\gamma}{h_{n-1}}\int_{\textbf{R}}P_n^2(y)\frac{w(y)}{(z-y)(y-t)}dyP_{n-1}(z),
\end{align}
where we have used the Christoffel-Darboux formula\cite{S1939}. Applying (\ref{1.7}), we arrive the desired result.
\end{proof}
\begin{corollary}
Two identities
\begin{equation}\label{2.6.1}
\gamma\int_{\textbf{R}}\frac{P_n^2(y)w(y)}{y-t}dy=\int_{\textbf{R}}P_n^2(y)v_0^{\prime}(y)w(y)dy
\end{equation}
and
\begin{equation}\label{2.6.2}
\frac{\gamma}{h_{n-1}}\int_{\textbf{R}}P_n(y)P_{n-1}(y)\frac{w(y)}{y-t}dy=\frac{1}{h_{n-1}}\int_{\textbf{R}}P_n(y)P_{n-1}(y)v_0^{\prime}(y)w(y)dy-n
\end{equation}
\end{corollary}
\begin{proof}
From the proof of theorem \ref{T1}, we see that
\begin{equation*}
\frac{1}{h_n}\int_{\textbf{R}}P_n^{\prime}(y)P_n(y)w(y)dy=0.
\end{equation*}
Applying integration by parts into the left side of above equation,
\begin{align*}
\frac{1}{h_n}\int_{\textbf{R}}P_n^{\prime}(y)P_n(y)w(y)dy=\frac{1}{h_n}\int_{\textbf{R}}P_n^2(y)v_0^{\prime}(y)w(y)dy-\frac{\gamma}{h_n}\int_{\textbf{R}}P_n^2(y)\frac{w(y)}{y-t}dy,
\end{align*}
which gives (\ref{2.6.1}).

The same steps to do below equation,
\begin{equation*}
\frac{1}{h_{n-1}}\int_{\textbf{R}}P_n^{\prime}(y)P_{n-1}(y)w(y)dy=n,
\end{equation*}
we obtain (\ref{2.6.2}).
\end{proof}

\begin{remark}
The equation (\ref{2.1}) is called a lowering operator, see\cite{CI1997,ZC2019,CI2015,FAZ2012,WZCMMA,ZBCZ,MC2019}. With the help of (\ref{2.1}) and (\ref{2.7}), $P_n(z)$ also satisfies the raising operator equation
\begin{equation*}
P_{n-1}^{\prime}(z)=-A_{n-1}(z)P_n(z)+(B_n(z)+v_0^{\prime}(z))P_{n-1}(z).
\end{equation*}
\end{remark}

Note that $A_n(z)$ and $B_n(z)$ are not independent, a direct calculation produces two fundamental supplementary conditions\cite{CI1997} valid for all $z$.
\begin{theorem}\label{T2}
The functions $A_n(z)$ and $B_n(z)$ satisfy the conditions:
\begin{equation}\label{2.7}
B_n(z)+B_{n+1}(z)=(z-\alpha_n)A_n(z)-v_0^{\prime}(z),
\end{equation}
\begin{equation}\label{2.8}
1+(z-\alpha_n)(B_{n+1}(z)-B_n(z))=\beta_{n+1}A_{n+1}(z)-\beta_nA_{n-1}(z).
\end{equation}
\end{theorem}
\begin{proof}
Using (\ref{1.7}), the recurrence relation $(z-\alpha_n)P_n(z,t,\gamma)=P_{n+1}(z,t,\gamma)+\beta_nP_{n-1}(z,t,\gamma)$ has a restatement, i. e.
\begin{equation}\label{2.9}
\frac{(z-\alpha_n)P_n(z,t,\gamma)}{h_n}=\frac{P_{n+1}(z,t,\gamma)}{h_n}+\frac{P_{n-1}(z,t,\gamma)}{h_{n-1}}.
\end{equation}

From the definitions of $A_n(z)$ and $B_n(z)$, (\ref{2.2}) and (\ref{2.3}), and substituting (\ref{2.9}) into below computation, we have
\begin{align}\label{2.10}
&B_n(z)+B_{n+1}(z)\nonumber\\
=&\int_{\textbf{R}}P_n(y)\bigg(\frac{P_{n+1}(y)}{h_n}+\frac{P_{n-1}(y)}{h_{n-1}}\bigg)\frac{v_0^{\prime}(z)-v_0^{\prime}(y)}{z-y}w(y)dy\nonumber\\
&+\gamma\int_{\textbf{R}}P_n(y)\bigg(\frac{P_{n+1}(y)}{h_n}+\frac{P_{n-1}(y)}{h_{n-1}}\bigg)\frac{w(y)}{(z-y)(y-t)}dy\nonumber\\
&=\frac{1}{h_n}\int_{\textbf{R}}(y-\alpha_n)P_n^2(y)\frac{v_0^{\prime}(z)-v_0^{\prime}(y)}{z-y}w(y)dy+\frac{\gamma}{h_n}\int_{\textbf{R}}\frac{(y-\alpha_n)P_n^2(y)w(y)}{(z-y)(y-t)}dy\nonumber\\
&=-\frac{1}{h_n}\int_{\textbf{R}}P_n^2(y)(v_0^{\prime}(z)-v_0^{\prime}(y))w(y)dy-\frac{\gamma}{h_n}\int_{\textbf{R}}\frac{P_n^2(y)w(y)}{y-t}dy\nonumber\\
&+\frac{1}{h_n}\int_{\textbf{R}}(z-\alpha_n)P_n^2(y)\frac{v_0^{\prime}(z)-v_0^{\prime}(y)}{z-y}w(y)dy+\frac{\gamma}{h_n}\int_{\textbf{R}}\frac{(z-\alpha_n)P_n^2(y)w(y)}{(z-y)(y-t)}dy\nonumber\\
&=(z-\alpha_n)A_n(z)-v_0^{\prime}(z)+\frac{1}{h_n}\int_{\textbf{R}}P_n^2(y)v_0^{\prime}(y)w(y)dy-\frac{\gamma}{h_n}\int_{\textbf{R}}\frac{P_n^2(y)w(y)}{y-t}dy,
\end{align}
where
\begin{equation*}
y-\alpha_n=(y-z)+(z-\alpha_n)
\end{equation*}
and (\ref{2.6.1}) are used.The proof of (\ref{2.7}) is completed.\\
Similarly steps, by using the definition of $A_n(z)$ and (\ref{1.7}), we have
\begin{align}\label{2.16}
&\beta_{n+1}A_{n+1}(z)-\beta_nA_{n-1}(z)\nonumber\\
=&\frac{1}{h_n}\int_{\textbf{R}}\frac{v_0^{\prime}(z)-v_0^{\prime}(y)}{z-y}P_{n+1}^2(y)w(y)dy+\frac{\gamma}{h_n}\int_{\textbf{R}}\frac{P_{n+1}^2(y)}{(z-y)(y-t)}w(y)dy\nonumber\\
&-\frac{\beta_n}{h_{n-1}}\int_{\textbf{R}}\frac{v_0^{\prime}(z)-v_0^{\prime}(y)}{z-y}P_{n-1}^2(y)w(y)dy-\frac{\gamma\beta_n}{h_{n-1}}\int_{\textbf{R}}\frac{P_{n-1}^2(y)w(y)}{(z-y)(y-t)}dy\nonumber\\
=&\int_{\textbf{R}}\frac{v_0^{\prime}(z)-v_0^{\prime}(y)}{z-y}\bigg(\frac{P_{n+1}^2(y)}{h_n}-\frac{\beta_nP_{n-1}^2(y)}{h_{n-1}}\bigg)w(y)dy\nonumber\\
&+\gamma\int_{\textbf{R}}\bigg(\frac{P_{n+1}^2(y)}{h_n}-\frac{\beta_nP_{n-1}^2(y)}{h_{n-1}}\bigg)\frac{w(y)}{(z-y)(y-t)}dy.
\end{align}
On the other hand,
Applying (\ref{2.6.2}) and (\ref{2.9}), we consider
\begin{align}\label{2.18}
&(z-\alpha_n)(B_{n+1}(z)-B_n(z))\nonumber\\
=&(z-\alpha_n)\bigg(\int_{\textbf{R}}\frac{v_0^{\prime}(z)-v_0^{\prime}(y)}{z-y}\bigg(\frac{P_{n+1}(y)}{h_n}-\frac{P_{n-1}(y)}{h_{n-1}}\bigg)P_n(y)w(y)dy\nonumber\\
&+\gamma\int_0^{\infty}\frac{P_n(y)}{(z-y)(y-t)}\bigg(\frac{P_{n+1}(y)}{h_n}-\frac{P_{n-1}(y)}{h_{n-1}}\bigg)w(y)dy\nonumber\\
=&\int_{\textbf{R}}(v_0^{\prime}(z)-v_0^{\prime}(y))\bigg(\frac{P_{n+1}(y)}{h_n}-\frac{P_{n-1}(y)}{h_{n-1}}\bigg)P_n(y)w(y)dy\nonumber\\
&+\int_{\textbf{R}}\frac{v_0^{\prime}(z)-v_0^{\prime}(y)}{z-y}\bigg(\frac{P_{n+1}^2(y)}{h_n^2}-\frac{P_{n-1}^2(y)}{h_{n-1}^2}\bigg)h_nw(y)dy\nonumber\\
&+\gamma\int_{\textbf{R}}\frac{P_n(y)}{y-t}\bigg(\frac{P_{n+1}(y)}{h_n}-\frac{P_{n-1}(y)}{h_{n-1}}\bigg)w(y)dy\nonumber\\
&+\gamma\int_{\textbf{R}}\frac{w(y)}{(z-y)(y-t)}\bigg(\frac{P_{n+1}^2(y)}{h_n^2}-\frac{P_{n-1}^2(y)}{h_{n-1}^2}\bigg)h_ndy\nonumber\\
=&-\bigg(\frac{1}{h_n}\int_{\textbf{R}}P_{n+1}(y)P_n(y)v_0^{\prime}(y)w(y)dy-\frac{\gamma}{h_{n}}\int_{\textbf{R}}\frac{P_n(y)P_{n+1}(y)w(y)}{y-t}dy\bigg)\nonumber\\
&+\bigg(\frac{1}{h_{n-1}}\int_{\textbf{R}}P_{n}(y)P_{n-1}(y)v_0^{\prime}(y)w(y)dy-\frac{\gamma}{h_{n-1}}\int_{\textbf{R}}\frac{P_{n-1}(y)P_n(y)w(y)}{y-t}dy\bigg)\nonumber\\
&+\int_{\textbf{R}}\frac{v_0^{\prime}(z)-v_0^{\prime}(y)}{z-y}\bigg(\frac{P_{n+1}^2(y)}{h_n}-\frac{\beta_nP_{n-1}^2(y)}{h_{n-1}}\bigg)w(y)dy\nonumber\\
&+\gamma\int_{\textbf{R}}\frac{w(y)}{(z-y)(y-t)}\bigg(\frac{P_{n+1}^2(y)}{h_n}-\frac{\beta_nP_{n-1}^2(y)}{h_{n-1}}\bigg)dy\nonumber\\
=&\int_{\textbf{R}}\frac{v_0^{\prime}(z)-v_0^{\prime}(y)}{z-y}\bigg(\frac{P_{n+1}^2(y)}{h_n}-\frac{\beta_nP_{n-1}^2(y)}{h_{n-1}}\bigg)w(y)dy\nonumber\\
&+\gamma\int_{\textbf{R}}\frac{w(y)}{(z-y)(y-t)}\bigg(\frac{P_{n+1}^2(y)}{h_n}-\frac{\beta_nP_{n-1}^2(y)}{h_{n-1}}\bigg)dy-1.
\end{align}
(\ref{2.16}) minus (\ref{2.18}) follows
\begin{equation*}
\beta_{n+1}A_{n+1}(z)-\beta_nA_{n-1}(z)-(z-\alpha_n)(B_{n+1}(z)-B_n(z))=1,
\end{equation*}
and we arrive (\ref{2.8}).
\end{proof}
\begin{remark}
(\ref{2.7}) and (\ref{2.8}) are the well-known supplementary conditions ($S_1$) and ($S_2$), respectively. ($S_1$) and ($S_2$) have been applied to random matrix theory in \cite{TW1994}.
\end{remark}
We produces an identity involving $\sum\limits_{k=0}^{n-1} A_k(z)$ by the combination of ($S_1$) and ($S_2$).
\begin{theorem}
$A_n(z)$, $B_n(z)$ and $\sum\limits_{k=0}^{n-1} A_k(z)$ satisfy the identity
\begin{equation}
B_n^2(z)+v_0^{\prime}(z)B_n(z)+\sum\limits_{k=0}^{n-1} A_k(z)=\beta_nA_n(z)A_{n-1}(z).\tag{$S_2^{\prime}$}
\end{equation}
\end{theorem}
\begin{proof}
Using ($S_1$) with multiplying ($S_2$) by $A_n(z)$ on both sides, we have
\begin{align*}
&A_n(z)+B_{n+1}^2(z)-B_n^2(z)+v_0^{\prime}(z)(B_{n+1}(z)-B_n(z))\\
=&\beta_{n+1}A_{n+1}(z)A_n(z)-\beta_nA_n(z)A_{n-1}(z).
\end{align*}
Taking a telescopic sum together with the appropriate "initial conditions", $B_0(z)=0$ and $\beta_0A_{-1}(z)=0$, we obtain the desired result.
\end{proof}

\begin{remark}
The compatibility conditions ($S_1$), ($S_2$) and ($S_2^{\prime}$) are valid for $z\in\textbf{C}\cup\{\infty\}$.

\end{remark}

Using Theorem \ref{T1} and Theorem \ref{T2}, we obtain the second order linear ordinary differential equation satisfied by $P_n(z)$.
\begin{theorem}
The monic orthogonal polynomials $P_n(z)$ satisfy the second order differential equation,
\begin{equation}\label{2.13}
P_n^{\prime\prime}(z)+Q_n(z)P_n^{\prime}(z)+T_n(z)P_n(z)=0
\end{equation}
where
\begin{equation}\label{2.14}
Q_n(z)=-v_0^{\prime}(z)-\frac{A_n^{\prime}(z)}{A_n(z)},
\end{equation}
\begin{equation}\label{2.15}
T_n(z)=B_n^{\prime}(z)-\frac{B_n(z)A_n^{\prime}(z)}{A_n(z)}+\beta_nA_n(z)A_{n-1}(z)-B_n^2(z)-B_n(z)v_0^{\prime}(z).
\end{equation}
\end{theorem}
\begin{proof}
See \cite{ZC2019,CFZ2019,DZY,MC2019}.
\end{proof}

\section{Deformed Hermite weight}
\noindent Recalling Theorem \ref{T1},
\begin{align*}
A_n(z)=\frac{1}{h_n}\int_{\textbf{R}}P_n^2(y)\frac{v_0^{\prime}(z)-v_0^{\prime}(y)}{z-y}w(y)dy+a_n(t),
\end{align*}
\begin{equation*}
B_n(z)=\frac{1}{h_{n-1}}\int_{\textbf{R}}P_n(y)P_{n-1}(y)\frac{v_0^{\prime}(z)-v_0^{\prime}(y)}{z-y}w(y)dy+b_n(t),
\end{equation*}
\begin{equation*}
a_n(t)=\frac{\gamma}{h_n}\int_{\textbf{R}}\frac{P_n^2(y)}{(z-y)(y-t)}w(y)dy
\end{equation*}
and
\begin{equation*}
b_n(t)=\frac{\gamma}{h_{n-1}}\int_{\textbf{R}}\frac{P_n(y)P_{n-1}(y)}{(z-y)(y-t)}w(y)dy,
\end{equation*}
we have the following expression of $A_n(z)$ and $B_n(z)$.
\begin{theorem}
As $n\rightarrow\infty$, it yields
\begin{equation}\label{3.1}
A_n(z)=2+\frac{R_n(t)}{z}+\frac{\gamma+tR_n(t)}{z^2}+\frac{\gamma\alpha_n+\gamma t+t^2R_n(t)}{z^3}+\mathcal{O}(z^{-4})
\end{equation}
and
\begin{equation}\label{3.2}
B_n(z)=\frac{r_n(t)}{z}+\frac{tr_n(t)}{z^2}+\frac{\gamma\beta_n+t^2r_n(t)}{z^3}+\mathcal{O}(z^{-4}),
\end{equation}
where
\begin{equation}\label{3.1.1}
R_n(t)=\frac{\gamma}{h_n}\int_{\mathbf{R}}\frac{P_n^2(y)w(y)}{y-t}dy
\end{equation}
and
\begin{equation}\label{3.2.1}
r_n(t)=\frac{\gamma}{h_{n-1}}\int_{\mathbf{R}}\frac{P_n(y)P_{n-1}(y)w(y)}{y-t}dy.
\end{equation}
\end{theorem}
\begin{proof}
As $n\rightarrow\infty$,
\begin{align*}
\frac{1}{z-y}&=\frac{1}{z}\cdot\frac{1}{1-\frac{y}{z}}=\frac{1}{z}\bigg(1+\frac{y}{z}+\frac{y^2}{z^2}+\mathcal{O}(z^{-3})\bigg)\\
&=\frac{1}{z}+\frac{y}{z^2}+\frac{y^2}{z^3}+\mathcal{O}(z^{-4}),
\end{align*}
we can rewrite $a_n(t)$ and $b_n(t)$, respectively,
\begin{align}\label{3.1.2}
a_n(t)=&\frac{\gamma}{zh_n}\int_{\mathbf{R}}\frac{P_n^2(y)w(y)}{y-t}dy+\frac{\gamma}{z^2h_n}\int_{\mathbf{R}}\frac{yP_n^2(y)w(y)}{y-t}dy+\frac{\gamma}{z^3h_n}\int_{\mathbf{R}}\frac{y^2P_n^2(y)w(y)}{y-t}dy\nonumber\\
&+\mathcal{O}(z^{-4}),
\end{align}
\begin{align}\label{3.2.2}
b_n(t)=&\frac{\gamma}{zh_{n-1}}\int_{\mathbf{R}}\frac{P_n(y)P_{n-1}(y)w(y)}{y-t}dy+\frac{\gamma}{z^2h_{n-1}}\int_{\mathbf{R}}\frac{yP_n(y)P_{n-1}(y)w(y)}{y-t}dy\nonumber\\
&+\frac{\gamma}{z^3h_{n-1}}\int_{\mathbf{R}}\frac{y^2P_n(y)P_{n-1}(y)w(y)}{y-t}dy+\mathcal{O}(z^{-4}).
\end{align}
Note
\begin{align}\label{3.1.3}
&\frac{\gamma}{h_n}\int_{\mathbf{R}}\frac{yP_n^2(y)w(y)}{y-t}dy=\frac{\gamma}{h_n}\int_{\mathbf{R}}\frac{((y-t)+t)P_n^2(y)w(y)}{y-t}dy\nonumber\\
=&\frac{\gamma}{h_n}\bigg(\int_{\mathbf{R}}P_n^2(y)w(y)dy+t\int_{\mathbf{R}}\frac{P_n^2(y)w(y)}{y-t}dy\bigg)\nonumber\\
=&\gamma+\frac{\gamma t}{h_n}\int_{\mathbf{R}}\frac{P_n^2(y)w(y)}{y-t}dy,
\end{align}
\begin{align}\label{3.1.4}
&\frac{\gamma}{h_n}\int_{\mathbf{R}}\frac{y^2P_n^2(y)w(y)}{y-t}dy=\frac{\gamma}{h_n}\int_{\mathbf{R}}\frac{((y^2-t^2)+t^2)P_n^2(y)w(y)}{y-t}dy\nonumber\\
=&\frac{\gamma}{h_n}\bigg(\int_{\mathbf{R}}(y+t)P_n^2(y)w(y)dy+t^2\int_{\mathbf{R}}\frac{P_n^2(y)w(y)}{y-t}dy\bigg)\nonumber\\
=&\frac{\gamma}{h_n}\bigg(\alpha_nh_n+th_n+t^2\int_{\mathbf{R}}\frac{P_n^2(y)w(y)}{y-t}dy\bigg)\nonumber\\
=&\gamma\alpha_n+\gamma t+\frac{\gamma t^2}{h_n}\int_{\textbf{R}}\frac{P_n^2(y)w(y)}{y-t}dy,
\end{align}
\begin{align}\label{3.2.3}
&\frac{\gamma}{h_{n-1}}\int_{\mathbf{R}}\frac{yP_n(y)P_{n-1}(y)w(y)}{y-t}dy=\frac{\gamma}{h_{n-1}}\int_{\mathbf{R}}\frac{((y-t)+t)P_n(y)P_{n-1}(y)w(y)}{y-t}dy\nonumber\\
=&\frac{\gamma}{h_{n-1}}\bigg(\int_{\mathbf{R}}P_n(y)P_{n-1}(y)w(y)dy+t\int_{\mathbf{R}}\frac{P_n(y)P_{n-1}(y)w(y)}{y-t}dy\bigg)\nonumber\\
=&\frac{\gamma t}{h_{n-1}}\int_{\mathbf{R}}\frac{P_n(y)P_{n-1}(y)w(y)}{y-t}dy,
\end{align}
and
\begin{align}\label{3.2.4}
&\frac{\gamma}{h_{n-1}}\int_{\mathbf{R}}\frac{y^2P_n(y)P_{n-1}(y)w(y)}{y-t}dy=\frac{\gamma}{h_{n-1}}\int_{\mathbf{R}}\frac{((y^2-t^2)+t^2)P_n(y)P_{n-1}(y)w(y)}{y-t}dy\nonumber\\
=&\frac{\gamma}{h_{n-1}}\bigg(\int_{\mathbf{R}}(y+t)P_n(y)P_{n-1}(y)w(y)dy+t^2\int_{\mathbf{R}}\frac{P_n(y)P_{n-1}(y)w(y)}{y-t}dy\bigg)\nonumber\\
=&\frac{\gamma}{h_{n-1}}\bigg(\beta_nh_{n-1}+t^2\int_{\mathbf{R}}\frac{P_n(y)P_{n-1}(y)w(y)}{y-t}dy\bigg)\nonumber\\
=&\gamma\beta_n+\frac{\gamma t^2}{h_{n-1}}\int_{\textbf{R}}\frac{P_n(y)P_{n-1}(y)w(y)}{y-t}dy,
\end{align}
where we used (\ref{1.4}).
Substituting (\ref{3.1.3}), (\ref{3.1.4}), (\ref{3.2.3}) and (\ref{3.2.4}) into (\ref{3.1.2}) and (\ref{3.2.2}), the theorem is established.
\end{proof}

Putting (\ref{3.1}) and (\ref{3.2}) into ($S_1$), and comparing the constants and coefficients of $\frac{1}{z}$, respectively, we have
\begin{equation}\label{3.3}
2\alpha_n=t+R_n(t),
\end{equation}
and
\begin{equation}\label{3.4}
r_{n+1}(t)+r_n(t)=\gamma+(t-\alpha_n)R_n(t).
\end{equation}
Carrying out a similar calculation with ($S_2$) and ($S_2^{\prime}$), they give rise to another three identities:
\begin{equation}\label{3.5}
1-r_n(t)+r_{n+1}(t)=2\beta_{n+1}-2\beta_n,
\end{equation}
\begin{equation}\label{3.6}
(t-\alpha_n)(r_{n+1}(t)-r_n(t))=\beta_{n+1}R_{n+1}(t)-\beta_nR_{n-1}(t)
\end{equation}
and
\begin{equation}\label{3.7}
\sum\limits_{j=0}^{n-1}R_j(t)=2\beta_n\left(R_{n-1}(t)+R_n(t)\right)-tr_n(t).
\end{equation}
A telescopic sum of (\ref{3.5}) gives
\begin{equation}\label{3.8}
n+r_n(t)=2\beta_n.
\end{equation}
Multiplying $R_n(t)$ on both sides of (\ref{3.6}),
\begin{equation*}
(t-\alpha_n)R_n(t)(r_{n+1}(t)-r_n(t))=\beta_{n+1}R_{n+1}(t)R_n(t)-\beta_nR_{n-1}(t)R_n(t),
\end{equation*}
then eliminating $(t-\alpha_n)R_n(t)$ by (\ref{3.4}), it yields
\begin{equation*}
(r_{n+1}^2(t)-r_n^2(t)-\gamma)(r_{n+1}(t)-r_n(t))=\beta_{n+1}R_{n+1}(t)R_n(t)-\beta_nR_{n-1}(t)R_n(t).
\end{equation*}
A telescopic sum of above equation,
\begin{equation}\label{3.9}
r_n^2(t)-\gamma r_n(t)=\beta_n R_{n-1}(t)R_n(t).
\end{equation}
\begin{rem}
(\ref{3.3}), (\ref{3.4}), (\ref{3.5}), (\ref{3.6}), (\ref{3.7}), (\ref{3.8}) and (\ref{3.9}) are generated from ($S_1$), ($S_2$) and ($S_2^{\prime}$), which are significant for the next discussions.
\end{rem}

\section{The $t$ dependance}
\noindent Mention again that the weight (\ref{1.1}) depends on $t$, it means all of the quantities considered in this paper can be viewed as functions in $t$. In this section, we will investigate their dependance with respect to their parameter.

Taking a derivative with respect to $t$ in (\ref{1.3}) when $i=j=n$,
\begin{equation*}
h_n(t,\gamma)=\int_{\mathbf{R}}P_n^2(z,t,\gamma)w(z.t,\gamma)dz=\int_{\mathbf{R}}P_n^2(z,t,\gamma)e^{-z^2+tz}|z-t|^{\gamma}(A+B\theta(z-t))dz,
\end{equation*}
we find
\begin{align}\label{4.1}
\frac{\partial h_n(t,\gamma)}{\partial t}&=\int_{\mathbf{R}}P_n^2(z,t,\gamma)w(z,t,\gamma)zdz+\int_{\mathbf{R}}\frac{\gamma P_n^2(z,t,\gamma)w(z,t,\gamma)}{t-z}dz\nonumber\\
&=\alpha_n h_n-\gamma \int_{\mathbf{R}}\frac{P_n^2(z,t,\gamma)w(z,t,\gamma)}{z-t}dz\nonumber\\
&=\alpha_n h_n-R_n(t)h_n,
\end{align}
where
\begin{equation*}
\frac{\partial(|z-t|^{\gamma})}{\partial t}=\frac{\partial(|t-z|^{\gamma})}{\partial t}=\bigg((t-z)^{\gamma}-(z-t)^{\gamma}\bigg)\delta(t-z)+\frac{\gamma|t-z|^{\gamma}}{t-z}
\end{equation*}
and the definition of $R_n(t)$ are used.
Straightforward,
\begin{equation}\label{4.2}
\frac{d\ln h_n(t)}{dt}=\frac{1}{h_n(t)}\frac{dh_n(t)}{dt}=\alpha_n-R_n(t)=\frac{t}{2}-\frac{R_n(t)}{2},
\end{equation}
where $\alpha_n$ is eliminated by (\ref{3.3}).

It follows
\begin{align*}
\frac{d\ln \beta_n}{dt}&=\frac{1}{h_n(t)}\frac{dh_n(t)}{dt}-\frac{1}{h_{n-1}(t)}\frac{dh_{n-1}(t)}{dt}\nonumber\\
&=\frac{1}{2}(R_{n-1}(t)-R_n(t)),
\end{align*}
where we used (\ref{1.7}),
and we can pose the derivative of $\beta_n$ by the difference between $R_{n-1}(t)$ and $R_n(t)$,
\begin{equation*}
\frac{d \beta_n}{dt}=\frac{\beta_n}{2}(R_{n-1}(t)-R_n(t)).
\end{equation*}

On the other hand, similarly steps, taking a derivative of
 \begin{equation*}
 \int_{\mathbf{R}}P_n(z,t,\gamma)P_{n-1}(z,t,\gamma)w(z,t,\gamma)dz=0
 \end{equation*}
 with respect to $t$ and using the definition of $r_n(t)$, it produces
\begin{align*}
0=&\int_{\mathbf{R}}\frac{\partial P_n(z,t,\gamma)}{\partial t}P_{n-1}(z,t,\gamma)w(z,t,\gamma)dz+\int_{\mathbf{R}}P_n(z,t,\gamma)P_{n-1}(z,t,\gamma)w(z,t,\gamma)zdz\\
&+\int_{\mathbf{R}}P_n(z,t,\gamma)P_{n-1}(z,t,\gamma)e^{-z^2+tz}\frac{\partial(|z-t|^{\gamma})}{\partial t}(A+B\theta(z-t))dz\\
&+\int_{\mathbf{R}}P_n(z,t,\gamma)P_{n-1}(z,t,\gamma)e^{-z^2+tz}|z-t|^{\gamma}\frac{\partial(A+B\theta(z-t))}{\partial t}dz\\
=&h_{n-1}\frac{d\textbf{p}(n,t)}{dt}+\beta_nh_{n-1}-\gamma\int_{\mathbf{R}}\frac{P_n(z)P_{n-1}(z)w(z)}{z-t}dz\\
=&h_{n-1}\frac{d\textbf{p}(n,t)}{dt}+\beta_nh_{n-1}-r_n(t)h_{n-1},
\end{align*}
simplify,
\begin{equation}\label{4.4}
\frac{d\textbf{p}(n,t)}{dt}=r_n(t)-\beta_n.
\end{equation}
Differential (\ref{1.6}) and applying (\ref{4.4}),
\begin{align*}
\alpha_n^{\prime}(t)&=r_n(t)-\beta_n-r_{n+1}(t)+\beta_{n+1}\nonumber\\
&=\frac{1}{2}(r_n(t)-r_{n+1}(t)+1),
\end{align*}
since (\ref{3.8}).
We summarize in the following theorem.
\begin{theorem}
The recurrence coefficients $\alpha_n$ and $\beta_n$ are expressed in terms of $r_n(t)$, $r_{n+1}(t)$, $R_{n-1}(t)$ and $R_n(t)$ as follows:
\begin{equation}\label{4.3}
\beta_n^{\prime}(t)=\frac{\beta_n}{2}\left(R_{n-1}(t)-R_n(t)\right),
\end{equation}
and
\begin{align}\label{4.5}
\alpha_n^{\prime}(t)&=\frac{1}{2}(r_n(t)-r_{n+1}(t)+1).
\end{align}
\end{theorem}
\begin{theorem}
The auxiliary quantities $r_n(t)$ and $R_n(t)$ satisfy the coupled Riccati equations
\begin{equation}\label{4.6}
r_n^{\prime}(t)=\frac{r_n^2(t)-\gamma r_n(t)}{R_n(t)}-\frac{(r_n(t)+n)R_n(t)}{2},
\end{equation}
\begin{equation}\label{4.7}
R_n^{\prime}(t)=2r_n(t)-\frac{1}{2}(t-R_n(t))R_n(t)-\gamma.
\end{equation}
\end{theorem}
\begin{proof}
Taking a derivative of (\ref{3.8}) with respect to $t$,
\begin{equation}\label{4.8}
2\beta_n^{\prime}(t)=r_n^{\prime}(t),
\end{equation}
eliminating $\beta_n^{\prime}(t)$ by (\ref{4.3}),
\begin{equation*}
r_{n}^{\prime}(t)=\beta_nR_{n-1}(t)-\beta_nR_n(t),
\end{equation*}
then eliminating $\beta_nR_{n-1}(t)$ by (\ref{3.9}), and eliminating another $\beta_n$ by (\ref{3.8}), respectively, we arrive (\ref{4.6}).

Combining (\ref{3.4}) and (\ref{3.3}), and eliminating $\alpha_n$,
\begin{equation}\label{4.9}
r_{n+1}(t)+r_n(t)=\gamma+\frac{1}{2}(t-R_n(t))R_n(t).
\end{equation}
Substituting (\ref{4.9}) into (\ref{4.5}), and eliminating $r_{n+1}(t)$,
\begin{equation}\label{4.10}
\alpha_n^{\prime}(t)=r_n(t)-\frac{1}{4}(t-R_n(t))R_n(t)+\frac{1-\gamma}{2}.
\end{equation}
We note the derivative on both sides of (\ref{3.3}),
\begin{equation}\label{4.11}
2\alpha_n^{\prime}(t)=1+R_n^{\prime}(t).
\end{equation}
A simple computation with combining (\ref{4.10}) and (\ref{4.11}) and eliminating $\alpha_n^{\prime}(t)$, (\ref{4.7}) is confirmed.
\end{proof}

\section{$\sigma-$Form}
 \noindent In this section, we introduce two auxiliary quantities $\sigma_n(t)$ and $\hat\sigma_n(t)$, and obtain their second order non-linear differential or difference equations. \\
\indent Define a quantity allied to the Hankel determinant,
 \begin{equation*}
 \sigma_n(t):=\frac{d}{dt}\ln D_n(t)=\frac{d}{dt}\sum\limits_{j=0}^{n-1}\ln h_j(t)=\sum\limits_{j=0}^{n-1}\frac{d}{dt}\left(\ln h_j(t)\right).
 \end{equation*}
Due to (\ref{4.2}), (\ref{3.3}) and (\ref{1.8}),
\begin{align}\label{3.10}
\sigma_n(t)&=\frac{nt}{2}-\frac{1}{2}\sum\limits_{j=0}^{n-1}R_j(t)=nt-\sum\limits_{j=0}^{n-1}\alpha_j=\textbf{p}(n,t,\gamma)+nt.
\end{align}
\begin{theorem}\label{T4}
The quantity $\sigma_n(t)$ satisfies
\begin{equation}\label{3.11}
(\sigma_n^{\prime\prime}(t))^2=\frac{1}{4}(t\sigma_n^{\prime}(t)-\sigma_n(t))^2-4(\sigma_n^{\prime}(t)-\frac{n}{2})\sigma_n^{\prime}(t)(\sigma_n^{\prime}(t)-\frac{n+\gamma}{2}),
\end{equation}
which is a particular Jimbo-Miwa-Okamoto $\sigma$-form of the Painlev\'{e} IV equation \cite{JM1981}.
\end{theorem}
\begin{proof}
In view of (\ref{3.10}),
\begin{equation}\label{3.11.1}
nt-2\sigma_n(t)=\sum\limits_{j=0}^{n-1}R_j(t).
\end{equation}
Combining (\ref{3.7}) and above equation,
\begin{equation*}
nt-2\sigma_n(t)=2\beta_nR_{n-1}(t)+2\beta_nR_n(t)-tr_n(t),
\end{equation*}
then eliminating $\beta_nR_{n-1}(t)$ and $\beta_n$ by (\ref{3.9}) and (\ref{3.8}), respectively,
\begin{equation}\label{3.12}
(n+r_n(t))R_n(t)+\frac{2(r_n^2(t)-\gamma r_n(t))}{R_n(t)}=t r_n(t)-2\sigma_n(t)+nt.
\end{equation}

On account of (\ref{4.6}),
\begin{equation}\label{3.13}
\frac{2(r_n^2(t)-\gamma r_n(t))}{R_n(t)}-(n+r_n(t))R_n(t)=2r_n^{\prime}(t).
\end{equation}
The sum of (\ref{3.12}) and (\ref{3.13}) gives
\begin{equation}\label{3.14}
\frac{4(r_n^2(t)-\gamma r_n(t))}{R_n(t)}=tr_n(t)-2\sigma_n(t)+nt+2r_n^{\prime}(t);
\end{equation}
oppositely, (\ref{3.12}) minus (\ref{3.13}),
\begin{equation}\label{3.15}
2(n+r_n(t))R_n(t)=tr_n(t)-2\sigma_n(t)+nt-2r_n^{\prime}(t).
\end{equation}
Then the product of (\ref{3.14}) and (\ref{3.15}) is
\begin{equation}\label{3.16}
8r_n(t)(r_n(t)+n)(r_n(t)-\gamma)=(tr_n(t)-2\sigma_n(t)+nt)^2-4(r_n^{\prime}(t))^2.
\end{equation}

Taking the derivative of (\ref{3.10}) with respect to $t$ gives to
\begin{equation}\label{3.17}
\sigma_n^{\prime}(t)=\textbf{p}^{\prime}(n,t,\gamma)+n,
\end{equation}
and eliminating $\beta_n$ between (\ref{3.8}) and (\ref{4.4}),
\begin{equation}\label{3.18}
\textbf{p}^{\prime}(n,t,\gamma)=\frac{r_n(t)}{2}-\frac{n}{2}.
\end{equation}
Inserting (\ref{3.18}) into (\ref{3.17}), and eliminating $\textbf{p}^{\prime}(n,t,\gamma)$,
\begin{equation}\label{3.19}
\sigma_n^{\prime}(t)=\frac{r_n(t)}{2}+\frac{n}{2}.
\end{equation}
Some simplifications after substituting (\ref{3.19}) into (\ref{3.16}), we get (\ref{3.11}) immediately.
\end{proof}
\begin{remark}
$\sigma_n(t)$ satisfies the continuous $\sigma-$form of the Painlev\'{e} IV, which theorem \ref{T4} is coincident with Tracy and Widom \cite{TW1994}.
\end{remark}
\begin{proposition}
We have
\begin{align}\label{3.25}
\sigma_n(t)=&\frac{d}{dt}\ln D_n(t)\nonumber\\
=&-\frac{(R_n^{\prime}(t))^2}{4R_n(t)}+\frac{R_n^3(t)}{16}-\frac{t}{8}R_n^2(t)+\bigg(\frac{t^2}{16}-\frac{n}{2}
-\frac{\gamma}{4}\bigg)R_n(t)+\frac{(2n+\gamma)t}{4}\nonumber\\
&+\frac{\gamma^2}{4R_n(t)}
\end{align}
\end{proposition}
\begin{proof}
For (\ref{3.7}), eliminating $\beta_nR_{n-1}(t)$ and $\beta_n$ by (\ref{3.9}) and (\ref{3.8}),
\begin{equation}\label{3.26}
\sum\limits_{j=0}^{n-1}R_j(t)=\frac{2(r_n^2(t)-\gamma r_n(t))}{R_n(t)}+(n+r_n(t))R_n(t)-tr_n(t).
\end{equation}
Substituting (\ref{3.26}) into (\ref{3.11.1}),
\begin{align}\label{3.27}
\sigma_n(t)=\frac{nt}{2}-\frac{r_n^2(t)-\gamma r_n(t)}{R_n(t)}-\frac{(n+r_n(t))R_n(t)}{2}+\frac{tr_n(t)}{2}.
\end{align}
Combining (\ref{4.7}) and (\ref{3.27}), and eliminating $r_n(t)$, the proof is completed.

\end{proof}

Also define the other quantity
\begin{equation}\label{3.20}
 \hat{\sigma}_n(t):=-\sum\limits_{j=0}^{n-1}R_j(t).
\end{equation}
\begin{theorem}
A second-order difference equation satisfied by $\hat{\sigma}_n(t)$,
\begin{align}\label{3.21}
&2(n\hat{\sigma}_{n-1}+\hat{\sigma}_{n}-n\hat{\sigma}_{n+1})^2-2\gamma(n\hat{\sigma}_{n-1}+\hat{\sigma}_n-n\hat{\sigma}_{n+1})(t-\hat{\sigma}_{n-1}+\hat{\sigma}_{n+1})\nonumber\\
=&(\hat{\sigma}_{n-1}-\hat{\sigma}_n)(\hat{\sigma}_n-\hat{\sigma}_{n+1})(\hat{\sigma}_n+nt)(\hat{\sigma}_{n+1}-\hat{\sigma}_{n-1}+t),
\end{align}
which is a discrete $\sigma$-form equation.
\end{theorem}
\begin{proof}
From (\ref{3.20}),
\begin{equation}\label{3.22}
R_n(t)=\hat{\sigma}_n(t)-\hat{\sigma}_{n+1}(t),
\end{equation}
we can rewrite (\ref{3.7}) with the aid of (\ref{3.8}) and (\ref{3.22}), i.e.
\begin{equation*}
(n+r_n(t))(\hat{\sigma}_{n-1}(t)-\hat{\sigma}_{n+1}(t))-tr_n(t)=-\hat{\sigma}_n(t),
\end{equation*}
solving $r_n(t)$,
\begin{equation}\label{3.23}
r_n(t)=\frac{n(\hat{\sigma}_{n-1}(t)-\hat{\sigma}_{n+1}(t))+\hat{\sigma}_n(t)}{t-\hat{\sigma}_{n-1}(t)+\hat{\sigma}_{n+1}(t)}.
\end{equation}
Similarly steps to recall (\ref{3.9}) by using (\ref{3.8}) and (\ref{3.22}),
\begin{equation}\label{3.24}
r_n^2(t)-\gamma r_n(t)=\frac{1}{2}(n+r_n(t))(\hat{\sigma}_{n-1}(t)-\hat{\sigma}_n(t))(\hat{\sigma}_n(t)-\hat{\sigma}_{n+1}(t)).
\end{equation}
Inserting (\ref{3.23}) into (\ref{3.24}) and simplify, it yields (\ref{3.21}).
\end{proof}
\begin{rem}
When $\gamma=0$, (\ref{3.21}) is a discrete $\sigma$-form of Painlev\'{e} \textbf{IV} equation, which corresponds to Min and Chen \cite{MC2019}.
\end{rem}

\section{Painlev\'{e} IV and asymptotic behaviors}
\begin{theorem}
The quantity $R_n(t)$ satisfies the following second order differential equation:
\begin{equation}\label{4.12}
R_n^{\prime\prime}(t)=\frac{(R_n^{\prime}(t))^2}{2R_n(t)}+\frac{3}{8}R_n^3(t)-\frac{t}{2}R_n^2(t)+\frac{1}{8}\left(t^2-8n-4-4\gamma\right)R_n(t)-\frac{\gamma^2}{2R_n(t)},
\end{equation}
which is a particular Painlev\'{e} IV equation. Moreover, let $\widetilde{R}_n:=\frac{1}{2}R_n(t)$ and $u:=-\frac{1}{4}t$, then
\begin{equation*}
\widetilde{R}_n^{\prime\prime}(u)=\frac{(\widetilde{R}_n^{\prime}(u))^2}{2\widetilde{R}_n(u)}+\frac{3}{2}\widetilde{R}_n^3(u)+4u\widetilde{R}_n^2(u)+2\left(u^2-\theta_1\right)\widetilde{R}_n(u)+\frac{\theta_2}{\widetilde{R}_n(u)},
\end{equation*}
which satisfies the Painlev\'{e} IV equation\cite{GLS2002} with $\theta_1=\frac{2n+1+\gamma}{4}$, $\theta_2=-\frac{\gamma^2}{8}$.
\end{theorem}
\begin{proof}
Solving for $r_n(t)$ from (\ref{4.7}),
\begin{equation*}
r_n(t)=\frac{1}{2}R_n^{\prime}(t)+\frac{1}{4}(t-R_n(t))R_n(t)+\frac{\gamma}{2},
\end{equation*}
and substituting the solution into (\ref{4.6}), simplify, we find (\ref{4.12}). \\
Let $\widetilde{R}_n:=\frac{1}{2}R_n(t)$ and $u:=-\frac{1}{4}t$, it is to see that $\widetilde{R}_n(u)$ takes the form
\begin{equation*}
\widetilde{R}_n^{\prime\prime}(u)=\frac{(\widetilde{R}_n^{\prime}(u))^2}{2\widetilde{R}_n(u)}+\frac{3}{2}\widetilde{R}_n^3(u)+4u\widetilde{R}_n^2(u)+2\left(u^2-\frac{n}{2}-\frac{1}{4}-\frac{\gamma}{4}\right)\widetilde{R}_n(u)-\frac{\gamma^2}{8\widetilde{R}_n(u)}.
\end{equation*}
\end{proof}
\begin{remark}
  A second-order differential equation satisfied by $r_n(t)$, which is related to the Chazy equation \cite{LC2017, MC2019}, can also be obtained by solving for $R_n(t)$ from (\ref{4.6}) and substituting the result of $R_n(t)$ into (\ref{4.7}). We do not write it down due to its too complicated.
\end{remark}

Disregard the derivative in (\ref{4.12}) and get a quartic equation
\begin{equation}\label{4.13}
\frac{3}{8}R_n^4(t)-\frac{t}{2}R_n^3(t)+\frac{1}{8}(t^2-8n-4-4\gamma)R_n^2(t)-\frac{\gamma^2}{2}=0.
\end{equation}
If $\gamma=0$ for (\ref{4.13}), $X_n(t)$ satisfies a quadratic equation,
\begin{equation*}
    3X_n^2(t)-4tX_n(t)+t^2-8n-4=0,
\end{equation*}
with the solutions
\begin{equation*}
    X_n(t)=\frac{2t\pm\sqrt{t^2+24n+12}}{3}.
\end{equation*}
Due to the definition of weight (\ref{1.1}), we separate two parts to discuss the asymptotic behaviors of $R_n(t)$.\\
(i)If $B>0$, we choose
\begin{equation*}
 X_n(t)=\frac{2t+\sqrt{t^2+24n+12}}{3},
\end{equation*}
as $n\rightarrow\infty$,
\begin{equation*}
X_n(t)=\frac{2\sqrt{6}}{3}n^{\frac{1}{2}}+\frac{2t}{3}+\frac{12+t^2}{12\sqrt{6}}n^{-\frac{1}{2}}-\frac{(12+t^2)^2}{1152\sqrt{6}}n^{-\frac{3}{2}}+\mathcal{O}(n^{-\frac{5}{2}}).
\end{equation*}
Hence, we assume the expansion of $R_n(t)$, as $n\rightarrow\infty$,
\begin{equation}\label{5.1}
R_n(t)=\sum\limits_{j=0}^{n-1}d_j(t)n^{\frac{1-j}{2}}.
\end{equation}
Substituting (\ref{5.1}) into (\ref{4.12}), we obtain
\begin{align*}
d_0(t)&=\frac{2\sqrt{6}}{3},~~~~d_1(t)=\frac{2t}{3},~~~~d_2(t)=\frac{\sqrt{6}(t^2+12\gamma+12)}{72},~~~~d_3(t)=0,\\
d_4(t)&=\frac{\sqrt{6}(288\gamma^2-24t^2\gamma-288\gamma-t^4-24t^2-240)}{6912},~~~~d_5(t)=\frac{(2-9\gamma^2)t}{72},
\end{align*}
which means
\begin{align*}
R_n(t)=&\frac{2\sqrt{6}}{3}n^{\frac{1}{2}}+\frac{2t}{3}+\frac{\sqrt{6}(t^2+12\gamma+12)}{72}n^{-\frac{1}{2}}\nonumber\\
&+\frac{\sqrt{6}(288\gamma^2-24t^2\gamma-288\gamma-t^4-24t^2-240)}{6912}n^{-\frac{3}{2}}+\frac{(2-9\gamma^2)t}{72}n^{-2}\nonumber\\
&+\mathcal{O}(n^{-\frac{5}{2}}).
\end{align*}
(ii)If $B<0$, we choose
\begin{equation*}
 X_n(t)=\frac{2t-\sqrt{t^2+24n+12}}{3},
\end{equation*}
as $n\rightarrow\infty$,
\begin{equation*}
X_n(t)=-\frac{2\sqrt{6}}{3}n^{\frac{1}{2}}+\frac{2t}{3}-\frac{12+t^2}{12\sqrt{6}}n^{-\frac{1}{2}}+\frac{(12+t^2)^2}{1152\sqrt{6}}n^{-\frac{3}{2}}+\mathcal{O}(n^{-\frac{5}{2}}).
\end{equation*}
Similarly, we assume the expansion of $R_n(t)$, as $n\rightarrow\infty$,
\begin{equation}\label{5.3}
R_n(t)=\sum\limits_{j=0}^{n-1}\widetilde{d}_j(t)n^{\frac{1-j}{2}}.
\end{equation}
Substituting (\ref{5.3}) into (\ref{4.12}), we obtain
\begin{align*}
\widetilde{d}_0(t)&=-\frac{2\sqrt{6}}{3},~~~~\widetilde{d}_1(t)=\frac{2t}{3},~~~~\widetilde{d}_2(t)=-\frac{\sqrt{6}(t^2+12\gamma+12)}{72},~~~~\widetilde{d}_3(t)=0,\\
\widetilde{d}_4(t)&=-\frac{\sqrt{6}(288\gamma^2-24t^2\gamma-288\gamma-t^4-24t^2-240)}{6912},~~~~\widetilde{d}_5(t)=\frac{(2-9\gamma^2)t}{72},
\end{align*}
which means
\begin{align*}
R_n(t)=&-\frac{2\sqrt{6}}{3}n^{\frac{1}{2}}+\frac{2t}{3}-\frac{\sqrt{6}(t^2+12\gamma+12)}{72}n^{-\frac{1}{2}}\nonumber\\
&-\frac{\sqrt{6}(288\gamma^2-24t^2\gamma-288\gamma-t^4-24t^2-240)}{6912}n^{-\frac{3}{2}}+\frac{(2-9\gamma^2)t}{72}n^{-2}\nonumber\\
&+\mathcal{O}(n^{-\frac{5}{2}}).
\end{align*}
Straightforward, we summarize the following item.
\begin{theorem}
As $n\rightarrow\infty$, the quantity $R_n(t)$ has the asymptotic expansions:
\begin{itemize}
  \item if $B>0$
  \begin{align}\label{5.2}
R_n(t)=&\frac{2\sqrt{6}}{3}n^{\frac{1}{2}}+\frac{2t}{3}+\frac{\sqrt{6}(t^2+12\gamma+12)}{72}n^{-\frac{1}{2}}\nonumber\\
&+\frac{\sqrt{6}(288\gamma^2-24t^2\gamma-288\gamma-t^4-24t^2-240)}{6912}n^{-\frac{3}{2}}+\frac{(2-9\gamma^2)t}{72}n^{-2}\nonumber\\
&+\mathcal{O}(n^{-\frac{5}{2}});
\end{align}
  \item if $B<0$
  \begin{align}\label{5.4}
R_n(t)=&-\frac{2\sqrt{6}}{3}n^{\frac{1}{2}}+\frac{2t}{3}-\frac{\sqrt{6}(t^2+12\gamma+12)}{72}n^{-\frac{1}{2}}\nonumber\\
&-\frac{\sqrt{6}(288\gamma^2-24t^2\gamma-288\gamma-t^4-24t^2-240)}{6912}n^{-\frac{3}{2}}+\frac{(2-9\gamma^2)t}{72}n^{-2}\nonumber\\
&+\mathcal{O}(n^{-\frac{5}{2}}).
\end{align}
\end{itemize}
\end{theorem}
By direct computations, we arrive at the following expansion formula for the scaled Hankel determinant
\begin{theorem}
For fixed $s>0$, $\gamma>-1$, we obtain
\begin{itemize}
  \item if $B>0$, as $n\rightarrow\infty$,
\begin{align}\label{5.5}
\ln\frac{D_n(s)}{D_n(0)}=&-\frac{2\sqrt{6}}{9}sn^{\frac{3}{2}}+\frac{s^2n}{12}-\frac{\sqrt{6}(s^3+36\gamma s)n^{\frac{1}{2}}}{216}+\frac{s^4}{864}+\frac{s^2\gamma}{24}\nonumber\\
&+\mathcal{O}(n^{-\frac{1}{2}})
\end{align}
  \item if $B<0$, as $n\rightarrow\infty$,
\begin{align}\label{5.6}
\ln\frac{D_n(s)}{D_n(0)}=&\frac{2\sqrt{6}}{9}sn^{\frac{3}{2}}+\frac{s^2n}{12}+\frac{\sqrt{6}(s^3+36\gamma s)n^{\frac{1}{2}}}{216}+\frac{s^4}{864}+\frac{s^2\gamma}{24}\nonumber\\
&+\mathcal{O}(n^{-\frac{1}{2}}).
\end{align}
\end{itemize}
\end{theorem}
\begin{proof}
Putting (\ref{5.2}) into (\ref{3.25}), the logarithmic derivative of the Hankel determinant $D_n(t)$ for $B>0$ becomes,
\begin{align}\label{5.7}
\frac{d}{dt}\ln D_n(t)=&-\frac{2\sqrt{6}}{9}n^{\frac{3}{2}}+\frac{nt}{6}-\frac{\sqrt{6}(t^2+12\gamma)n^{\frac{1}{2}}}{72}+\frac{t^3}{216}+\frac{t\gamma}{12}\nonumber\\
&+\mathcal{O}(n^{-\frac{1}{2}}).
\end{align}
Then, integration on both sides of equation (\ref{5.7}) from 0 to $s$,
\begin{equation*}
\int_0^{s}\frac{d}{dt}\ln D_n(t)dt=\ln D_n(s)-\ln D_n(0),
\end{equation*}
we obtain (\ref{5.5}). The proof is similar for $B<0$.
\end{proof}

Once the asymptotic behaviors for $R_n(t)$ is obtained, we also consider the large $n$ behavior of the monic orthogonal polynomials $P_n(z)$.
\begin{theorem}
As $n\rightarrow\infty$ and $t\rightarrow \infty$, $P_n(z)$ satisfies the second order differential-difference equation ,
\begin{equation}\label{5.8}
P_n^{\prime\prime}(z)+\left(\frac{3}{z}-2z+t\right)P_n^{\prime}(z)+\bigg(\left(\frac{2t^2}{3}+\frac{\gamma}{12}\right)\frac{n^2\gamma}{z^6}+\frac{2n^2\gamma t}{3z^5}+\frac{2n^2\gamma}{3z^4}\bigg)P_n(z)=0.
\end{equation}
\end{theorem}
\begin{proof}
Substituting (\ref{3.1}) and (\ref{3.2}) into (\ref{2.14}) and (\ref{2.15}), $Q_n(z)$ and $T_n(z)$ can be expressed in terms of $r_n(t)$, $R_n(t)$ and $\beta_n$. It should be pointed out that relationship among $R_n(t)$, $r_n(t)$, $\alpha_n$ and $\beta_n$ are shown by (\ref{3.3}), (\ref{3.8}) and (\ref{4.7}). The coefficients of (\ref{2.13}) are given in terms of $R_n(t)$. Sending $n\rightarrow\infty$, $t\rightarrow\infty$, and combining the asymptotic value (\ref{5.2}) or (\ref{5.4}), we obtain (\ref{5.8}).
\end{proof}
\begin{remark}
When $\gamma=0$, (\ref{5.8}) satisfies the bi-confluent Heun equation. More details of the Heun equation, see \cite{CFZ2019,H1888,R1995,SL2000,DZY,SK2010}.
\end{remark}

\section{Acknowledgements}
\noindent D. Wang, M. Zhu and Y. Chen would like to give thanks the Science and Technology Development Fund of the Macau SAR for providing FDCT 023/ 2017/A1. They would also like to thank the University of Macau for MYRG 2018-00125 FST.

\section{Availability of data}
\noindent The data that support the findings of this study are available from the corresponding author upon reasonable request.

\end{document}